\documentclass[aps,reprint,pra,nofootinbib,superscriptaddress,amsmath,amssymb]{revtex4-2}
\usepackage[utf8]{inputenc}
\usepackage{amsmath, amssymb, amsfonts,amsthm,amsopn}
\usepackage{dsfont}
\usepackage{graphicx}
\usepackage{thm-restate}
\usepackage{kbordermatrix}
\usepackage[normalem]{ulem}
\usepackage{array}
\usepackage{amssymb}

\usepackage{mathtools}
\usepackage{xcolor}

\usepackage[breaklinks]{hyperref}
\hypersetup{
    colorlinks,
    linkcolor={red!40!black},
    citecolor={blue!50!black},
    urlcolor={blue!20!black}
}

\DeclareMathOperator{\tr}{tr}
\DeclareMathOperator{\wt}{wt}

\DeclareMathOperator{\swap}{\textup{SWAP}}

\DeclareMathOperator{\diag}{diag}
\DeclareMathOperator{\supp}{supp}

\newcommand{\bra}[1]{\mathinner{\langle #1|}}
\newcommand{\ket}[1]{\mathinner{|#1\rangle}}
\newcommand{\braket}[2]{\mathinner{\langle #1|#2\rangle}}

\newcommand{\dyad}[1]{| #1\rangle \langle #1|}
\newcommand{\ot}[0]{\otimes}

\newcommand{\one}[0]{\mathds{1}}
\renewcommand{\a}{\alpha}
\renewcommand{\b}{\beta}

\newcommand{\N}{\mathds{N}}
\newcommand{\R}{\mathds{R}}
\newcommand{\C}{\mathds{C}}

\newcommand{\PP}{\mathcal{P}}

\newcommand{\bGamma}{\overbar{\Gamma}}
\newcommand{\eEight}{\Phi_{\textup{E8}}}

\newcommand{\av}[1]{\langle #1 \rangle}

\newtheorem{theorem}{Theorem}
\newtheorem*{theorem*}{Theorem}
\newtheorem{proposition}[theorem]{Proposition}

\newtheorem{corollary}[theorem]{Corollary}
\newtheorem{observation}[theorem]{Observation}

\newtheorem{remark}[theorem]{Remark}

\newtheorem*{problem*}{Problem}
\newtheorem*{question*}{Question}

\newtheorem*{result*}{Result}

\newcommand{\nn}{\nonumber}
\newcommand{\overbar}[1]{\mkern 1.5mu\overline{\mkern-1.5mu#1\mkern-1.5mu}\mkern 1.5mu}

\begin{document}

\title[
Combining moment matrices,
symmetric extension, and Lovász theta:\\
$\Phi_{\text{E8}}$ is entangled]{
Combining moment matrices,
symmetric extension, and Lovász theta:\\
$\Phi_{\text{E8}}$ is entangled}

\author{J\c{e}drzej Stempin}
\affiliation{
Division of Quantum Information,
Faculty of Mathematics, Physics and Informatics,
University of Gda{\'n}sk,
Wita Stwosza 57, 80-308 Gda{\'n}sk, Poland
}
\author{Gerard Anglès Munné}
\affiliation{
Institute of Theoretical Physics and Astrophysics, Faculty of Mathematics, Physics and Informatics,
University of Gda{\'n}sk,
Wita Stwosza 57, 80-308 Gda{\'n}sk, Poland
}

\author{Santiago Llorens}
\author{Felix Huber}
\affiliation{
Division of Quantum Information,
Faculty of Mathematics, Physics and Informatics,
University of Gda{\'n}sk,
Wita Stwosza 57, 80-308 Gda{\'n}sk, Poland
}
\email{felix.huber@ug.edu.pl}

\date{\today}

\begin{abstract}
We solve an open problem in entanglement theory posed by Yu et al., {\it Nature Communications 12, 1012 (2021)}.
The problem is to show,
via an entanglement witness,
that the $14$-qubit state $\Phi_{\text{E8}}$ is entangled.
Inspired by a method from quantum codes,
we combine symmetric extension
with moment matrices
to prove that $\eEight$ is entangled.
The proof has the form of a
rational infeasibility certificate
for a semidefinite program,
yielding an explicit entanglement witness.
Our approach unifies and extends several earlier methods
that involve the Lovász theta number of the Pauli anti-commutativity graph,
promising scalability and
flexibility in further applications.
\end{abstract}

\maketitle

\noindent{\bf Introduction. --- }
Determining whether a state is separable or entangled is a notoriously hard problem \cite{GURVITS2004448,Gharibian:2010:SNQ:2011350.2011361}.
At the same time, entanglement is widely thought to be a key ingredient in quantum computation~\cite{
10.1098/rspa.2002.1097,
PhysRevLett.131.030601,
PhysRevA.72.042316}.

In Ref.~\cite{Yu2021}, Yu et al. introduced the $14$-qubit state $\eEight$ \footnote{We denote this state by $\eEight$ as it appears in Eq.~(E8) in the arXiv version Ref.~\cite{Yu2021}, where it was first introduced. In the published version the state is in Eq. $(63)$.},
with the open task to prove
(via a witness)
that the state is entangled
across the bipartition
$A_1\dots A_7\,|\,B_1 \dots B_7$.
Curiously, it is known that
$\eEight$ {\it must} be entangled,
due to the fact that it serves as a proxy for a second (and equivalent)
problem in entanglement theory:
does there exist a pure seven-qubit state, such that all three-body marginals are maximally mixed?
Such states are known as {\it absolutely maximally entangled} (AME) \cite{PhysRevA.86.052335, PhysRevA.92.032316, helwig2013absolutelymaximallyentangledstates,
0034-4885-67-3-R03}.
For seven qubits~\cite{PhysRevLett.118.200502} no such state exists,
proving, albeit indirectly, entanglement in $\eEight$.

Indeed Yu et al. have shown that the problem of determining whether an AME state exists for a given number of qudits can be reformulated as a separability problem on two copies of the system \cite[Theorem 3]{Yu2021}:
a seven-qubit AME state exists if and only if the state $\eEight$ is separable, where~\cite[Supplemental Material, Eq.~(63)]{Yu2021},
\begin{equation}
  \eEight =  \sum_{i=0}^{3} x_{i}\,\mathcal{P}\big\{ P_{+}^{\otimes (7 - 2i)} \otimes P_{-}^{\otimes 2i} \big\}\,.
  \label{eq:phie8-state}
\end{equation}
Here $P_{\pm} = \frac{1}{2}(I \pm \text{SWAP}_{A_{i},B_{i}})$ denotes the projector onto the symmetric (+) and antisymmetric (-) subspaces of the qubit pair $A_{i}, B_{i}$, and $\mathcal{P}\{\cdot\}$ represents the sum over all permutations of the seven tensor factors that yield distinct terms. Here, we label the first set of seven qubits by $A_{1}, \ldots, A_{7}$ and the second set by $B_{1}, \ldots, B_{7}$. The coefficients are $x_{0} = \frac{113}{1119744}$, $x_{1} = \frac{17}{124416}$, $x_{2} = \frac{1}{13824}$, and $x_{3} = \frac{1}{1536}$.
The state is symmetric under the exchange of any $A_i$ and $B_i$  pair of tensor factors,
that is for all
$S\subseteq \{1,\ldots,7\}$ it holds that
$\pi_S \eEight\pi_S ^{\dag} = \eEight$
with
$\pi_S = \bigotimes_{i\in S}\text{SWAP}_{A_{i}B_{i}}$,
where $\text{SWAP}$ exchanges two tensor factors as
$\text{SWAP} \ket{v} \ot \ket{w} = \ket{w} \ot \ket{v}$~\footnote{We later also use $(ij)$ for the swap exchanging tensor factors $i,j$.}.
Also,
it is known that
for all $S\subseteq \{1,\dots, 7\}$,
\begin{equation}
\label{eq:SwapExp}
    \tr\big[\pi_S \eEight\big]
    =
    \frac{1}
    {\min\{2^{|S|}, 2^{7-|S|}
     \}}
     \,.
\end{equation}
In particular, %
$\tr\big[\pi_{\{1,\dots, 7\}}
\eEight\big]  = 1$.
Thus intuitively, $\eEight$ can be understood as the state obtained by twirling two copies of a (hypothetical) seven-qubit AME state under local unitaries.

\begin{figure}[tbp]
\includegraphics[width=0.98\columnwidth]{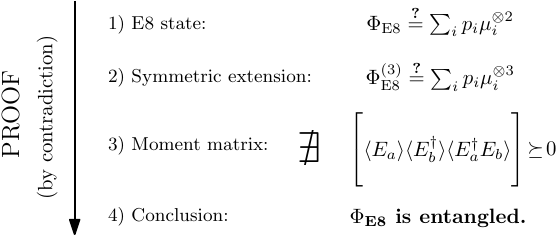}
\caption{
The putative separable state $\eEight$
allows for a symmetric three-copy extension
$\eEight^{(3)}$,
resulting in a highly constrained moment matrix that is nonlinear in expectations.
Via semidefinite programming (SDP) duality we show that no such moment matrix exists,
proving entanglement in $\eEight$.
}
\end{figure}

The open problem posed by Yu et al.~\cite{Yu2021} is to detect the entanglement of $\eEight$ directly, ideally via an entanglement witness\footnote{
Ref.~\cite{Yu2021}: "As no AME(7,2) state exists, the states
$\Phi_{AB}$
and
$\Phi^T_{AB}$
in Eqs. (E8, E9) are entangled, but we are not aware of any
operational entanglement criterion detecting them."}.
In particular, it is known that
it has a positive partial transpose,
and other entanglement criteria (see Ref.~\cite[Supplementary Note 5]{Yu2021})
have not detected entanglement
in this state.
In particular, the
symmetric extension hierarchy on density matrix level
with six copies does not detect its entanglement.

The aim of this paper is to prove entanglement in $\eEight$
by combining two  key numerical methods from quantum information:
the symmetric extension hierarchy by Doherty, Parrilo, and Spedalieri~\cite{PhysRevA.69.022308} for entanglement detection;
and the non-commutative optimization hierarchy
by Navascués, Pironio, and Acín \cite{navascues2007bounding, Navascues_2008},
traditionally reserved
for studying Bell nonlocality and ground state energies.

In our work we construct a symmetric extension of the putative separable state $\eEight$,
from which we derive a moment matrix.
A symmetry-reduced semidefinite program (SDP) allows to show
that no moment matrix satisfying all the constraints arising from $\eEight$ exists.
Via duality theory and an exact rational infeasibility certificate,
this provides a rigorous proof of entanglement in
$\eEight$.
This combination of methods --- symmetric extension, moment matrices, and symmetry-reduction ---
circumvents the prohibitive scaling inherent in symmetric extension based on density matrices alone.

The second contribution of this paper is to establish a link to coding bounds and graph theory.
That is,
we relate entanglement in $\eEight$ to the Lovász theta number,
a graph invariant
that was foundational in the study of semidefinite programs  for combinatorial optimization problems~\cite{
lovasz1979shannon,
Knuth_1994,
GALLI2017159}
and recently had a significant impact the study of
uncertainty relations~\cite{
PhysRevLett.132.200202,
PhysRevA.107.062211},
quantum correlations~\cite{
PhysRevA.107.062211},
stabilizer testing~\cite{10.1145/3717823.3718201},
ground state energies~\cite{hastings2023optimizing,
cbqf-d24r,
huber2025lovaszthetalowerbound},
and
quantum coding theory~\cite{6319408}.
Indeed our method is directly related to SDP bounds
on quantum codes developed in Refs.~\cite{
munne2024sdpboundsquantumcodes,
munne2026sdpboundsquantumcodesrational}.

\smallskip
\noindent {\bf Matrix inequality and proof sketch. --- }
We first require the following Lemma by Rains~\cite{959270}, which for the convenience of the reader we reproduce here.
For this, let $(12)$ (or $\swap$) be the swap operator for which
$(12) \ket{v} \ot \ket{w} = \ket{w} \ot \ket{v}$ for all $\ket{v}, \ket{w}\in~\C^d$.
Also denote by $(\cdot)^{T_{1}}$ the partial transposition on the first subsystem.

\begin{restatable}[Rains, Lemma 7.2~\cite{959270}]{lemma}{RainsLemma}
\label{Lem:Rains}

 Let $A$ be a complex square matrix
 of size $d \times d$. Then
 \begin{equation}\label{eq:AAtranspose}
     \big[(12) (A \ot A^\dag)\big]^{T_1}  \succeq 0\,,
 \end{equation}
 where $\succeq$ denotes positive semidefinite.
\end{restatable}
We provide a proof
in Appendix~\ref{app:aux},
and for now only note that
Eq.~\eqref{eq:AAtranspose}
can be
written as
$d(I \ot A) \dyad{\phi^+} (I \ot A^\dag) \succeq 0$,
where $\ket{\phi^+}$ is the Bell state of local dimension~$d$.

Let us shortly sketch the proof of our main result,
which is based on contradiction:
First,
symmetrically extend $\eEight$
to a three-partite state
$\eEight^{(3)}$.
Second, use Lemma~\ref{Lem:Rains}
to
map $\eEight^{(3)}$
to a positive semidefinite operator.
Third, derive a moment matrix $\Gamma$
with certain properties from it.
Fourth, show with a semidefinite program
that no such moment matrix exists.
This step requires to make the semidefinite program numerically approachable,  which we do via a symmetry reduction.
Despite being of numerical nature,
this last step will also be made rigorous
by the use of a solver with heuristic rounding to rational numbers.
This yields an exact rational infeasibility certificate,
proving entanglement in~$\eEight$.

\smallskip
\noindent{\bf A moment matrix witness. --- }
We now develop the main machinery to show entanglement in $\eEight$, a moment matrix witness based on symmetric extension.
For this, denote by $\mathcal{E}_n$ the n-qubit Pauli basis,
and let $E_0 \in \mathcal{E}_n$
be the identity matrix~$\one$. As we are concerned with
a bipartite state for which each party contains $7$ qubits,
set $N =  |\mathcal{E}_7 \,\backslash\, \one| = 4^7-1$.

\begin{observation}\label{obs:SDP}
If the state $\eEight$ is separable,
then the following semidefinite program (SDP) is feasible,
\begin{align}
\label{eq:SDP}
\text{find} \quad &\overbar{\Gamma} \succeq 0 \nn\\
\text{subject to} \quad
&\overbar{\Gamma}_{00} = 1\,,\nn\\
&\overbar{\Gamma}_{a0} = \overbar{\Gamma}_{aa}\,,\nn\\
&\overbar{\Gamma}_{ab} = 0 \quad\quad\,
\text{if} \quad\quad E_a^{\dag} E_b = -E_b E_a^{\dag}\,, \nn\\
&\overbar{\Gamma}_{aa} =
\tr\big( (E_a \otimes E_a^\dag)
   \eEight \big)\,,
\end{align}
with all relations to hold for $a,b=0,\dots, N$.
\end{observation}
\begin{proof}
Suppose the state $\eEight$ is separable.
Due to $\tr(\swap_{AB} \eEight)=1$ [Eq.~\eqref{eq:SwapExp}],
it can be written
as a symmetric\footnote{That the decomposition in Eq.~\eqref{eq:phi_sep} can be chosen symmetric follows
from the fact that $\tr(\pi \eEight)=1$~\cite{Yu2021}: suppose one writes the more general decomposition
$$
 \eEight = \sum_i p_i \dyad{\phi_i} \ot \dyad{\psi_i}\,.
$$
By $\tr(\pi \eEight)=1$
it must hold that
$|\!\braket{\psi_i}{\phi_i}\!|^2 = 1$, making
$\ket{\psi_i}, \ket{\phi_i}$
equal up to a complex phase.}
convex combination~\cite{Yu2021},
\begin{equation}
\label{eq:phi_sep}
    \eEight = \sum_{i}
    p_{i} \mu_i
    \otimes \mu_i\,,
\end{equation}
where $p_{i} \geq 0$ and $\sum_{i}p_{i}=1$.
As a consequence, there also exists a three-partite symmetric extension,
\begin{equation}
\label{eq:sym_ext}
\eEight^{(3)} =
\sum_i p_i \mu_i \otimes \mu_i \otimes \mu_i\,.
\end{equation}

From $\eEight^{(3)}$ we now construct a moment matrix.
As a first step, note that by Lemma~\ref{Lem:Rains}
the following operator is positive semidefinite,
\begin{equation}
    \label{eq:PsiDefinition}
    \Psi = \big((13)\eEight^{(3)}\big)^{T_{1}}\,.
\end{equation}
Define the set of operators
\begin{equation}\label{eq:R_a_ops}
    R_{a} = \frac{1}{2}(13)^{T_{1}}(I \ot E_{a} \ot E_{a}^{\dag}), \quad E_{a} \in \mathcal{E}_{7}\,,
\end{equation}
where $\mathcal{E}_{7}$ is the set of all Pauli strings of length $7$.
 Then construct a moment matrix $\Gamma$ with entries\footnote{Note that this construction was already given in the PhD thesis of one of the authors \cite[Eq. $(3.104)$]{munne2025existence},
up to misplaced daggers.}
\begin{equation}\label{eq:Gamma_first_def}
    \Gamma_{ab} = \tr(R_{a}^{\dag}R_{b}\Psi)\,.
\end{equation}
Note that the matrix $\Gamma$
is positive semidefinite,
\begin{equation}
\label{eq:Gamma_pos}
\Gamma \succeq 0\,,
\end{equation}
and has entries,
\begin{equation}\label{eq:Gamma_entries}
 \Gamma_{ab} = \sum_i p_i  \tr(E_{b}^{\dag} \mu_{i})\tr( E_{a}\mu_{i})\tr( E_{a}^{\dag}E_{b} \mu_{i}) \,.
\end{equation}
We derive Eqs.~\eqref{eq:Gamma_pos} and \eqref{eq:Gamma_entries} more carefully in
Appendix~\ref{app:gamma}.
In particular, from Eq.~\eqref{eq:Gamma_entries}
it follows that the diagonal of $\Gamma$ is determined by the correlations of $\eEight$,
\begin{align}\label{eq:Gamma_correl}
 \Gamma_{aa} &= \sum_i p_i
 \tr( E_a \mu_{i}) \tr( E_a^{\dag} \mu_{i})
 \nn\\
 &= \tr\big((E_a \ot E_a^{\dag}) \eEight\big) \,,
\end{align}
whereas the off-diagonal terms originate in correlations of the symmetric extension $\eEight^{(3)}$.
Additionally, from  Eq.~\eqref{eq:Gamma_entries} it follows due to $E_0 = \one$ that
\begin{align}\label{eq:Gamma_norm}
  \Gamma_{00} &= 1\,,\quad\quad
  \Gamma_{aa} = \Gamma_{a0}\,.
\end{align}
Finally, we restrict ourselves to the real part
 $\overbar{\Gamma} =  \big(\Gamma + \Gamma^T \big)/2$\,,
which satisfies the previous relations
Eqs.~\eqref{eq:Gamma_pos},
\eqref{eq:Gamma_correl},
\eqref{eq:Gamma_norm},
and also the following,
\begin{align}\label{eq:Gamma_acomm}
\overbar{\Gamma}_{ab} &= 0\quad\quad \text{if }\quad E_a^{\dag} E_b = -E_b E_a^{\dag}\,.
\end{align}

Collecting the constraints of Eqs.~\eqref{eq:Gamma_pos},
\eqref{eq:Gamma_correl}-\eqref{eq:Gamma_acomm}
yields the SDP
of Eq.~\eqref{eq:SDP}.
This ends the proof.
\end{proof}

We mention that the SDP in Eq.~\eqref{eq:SDP}
is a variant of
~\cite[Eq.~(145)]{munne2024sdpboundsquantumcodes}. We explain the relation in Appendix~\ref{app:relation}.

The values of
$\tr\big((E_a \ot E_a^\dagger)\eEight \big)$
for $E \in \mathcal{E}_7$ that are used as constraints
in Eq.~\eqref{eq:SDP}
are derived in Appendix~\ref{app:correlations} using the decomposition of the swap,
\begin{equation}\label{eq:swap_expansion}
    \swap = \frac{1}{2^{7}}
    \sum_{a=0}^N
    E_a\otimes E_a^\dag\,.
\end{equation}

While the derivation of $\Gamma$ through
Eqs.~\eqref{eq:sym_ext} - \eqref{eq:Gamma_first_def} makes
the connection to symmetric extension
via $\eEight^{(3)}$ clear,
the role of Lemma~\ref{Lem:Rains} appears to be rather obscure.
A more intuitive approach is to
understand $\Gamma$ as a convex combination of moment matrices of
one-partite states.
Suppose that there is only a single $\mu_i = \mu$ in the convex combination of Eq.~\eqref{eq:phi_sep},
and consider a set of observables of the form
$\mathfrak{E}_a = \langle E_a^\dag \rangle E_a$ where $\langle E_a \rangle = \tr(E_a\mu)$.
Then $\Gamma$ of Eq.~\eqref{eq:Gamma_first_def} reads
\begin{equation}\label{eq:Gamma_compact}
    \Gamma_{ab} = \tr(
{\mathfrak{E}_a}^\dag \mathfrak{E}_b \mu)\,.
\end{equation}
Thus $\Gamma$ has the form,
\begin{align}
	\kbordermatrix{
		& \one    & \av{E_1} E_1^\dagger & \cdots    & \av{E_{N}} E_{N}^\dagger \\
		\one &  1   & \Gamma_{01}           &\cdots     & \Gamma_{0 N} \\
		\av{E_1^\dagger} E_1 & & \Gamma_{11}    & \hdots   & \Gamma_{1N}\\
		\vdots &       &   & \ddots    & \vdots\\
		\av{E_{N}^\dagger} E_{N}& &           &           & \Gamma_{NN}
	} \succeq 0 \,,
	\label{eq:lovasz-moments}
\end{align}
with $\Gamma_{aa} = \Gamma_{a0} = \av{E_a^\dag}\av{E_a}$
and
$\Gamma_{00} = 1$
due to $E_0 = \one$,
from which the properties in
Eqs.~\eqref{eq:Gamma_correl}, \eqref{eq:Gamma_norm}
are easy to see
(also see Appendix~\ref{app:MomentMatrixFormulation}).
Taking the real part leads to
$\overbar{\Gamma}$.
For an arbitrary convex combination,
the resulting $\overbar{\Gamma}$ then is a convex combination of moment matrices,
$\overbar{\Gamma}=\sum_{i}p_i\overbar{\Gamma}_i$,
where $\overbar{\Gamma}_i$ is the moment matrix of Eq.~\eqref{eq:lovasz-moments} corresponding to a single state
$\mu_i$ in Eq.~\eqref{eq:phi_sep}.

\smallskip
\noindent{\bf The state $\eEight$ is entangled. --- }
We show that the SDP in Eq.~\eqref{eq:SDP} has no solution, thereby proving that $\eEight$ is entangled.
The challenge is that its matrix variable
has size $O(4^n) \times O(4^n)$,
and solving such SDP for $n=7$ is numerically prohibitive.

We address this challenge by formulating a {\it relaxation}
of the program in Eq.~\eqref{eq:SDP}
that
averages $\bGamma$
over a set of permutations
that
keep
weights and relative
"distances" of any triple
$E_a$,
$E_b$, and $E_a^\dag E_b$
invariant.
This yields a matrix $\widetilde \Gamma$.
This averaging procedure
induces symmetries
and
we can block-diagonalize (or {\it symmetry-reduce})
$\widetilde \Gamma$.
This is done through the Terwilliger algebra~\cite{bannai_algebraic_2021,1468304}.
Conceptually, because
the averaged (and later symmetry-reduced)
SDP is a relaxation,
infeasibility of this latter SDP
signals infeasibility of the original SDP in Eq.~\eqref{eq:SDP}.
The computations are rather tedious and we refer to Appendix~\ref{app:inf_cert}.

The formal proof of infeasibility is
then obtained
via duality theory of semidefinite programming~\cite{Vandenberghe1996}.
The idea is that any SDP in standard form,
\begin{equation}
\label{eq:PrimalSDPdef}
    \min_{X \in \R^{n \times n}}
    \Big\{\tr(CX) : \tr(A_i X) = b_i,\,
    i=1,\dots,m\,,\, X\succeq 0\Big\}\,,
\end{equation}
where $C \in \R^{n\times n}$,
has a dual SDP of the form,
\begin{equation}
\label{eq:DualSDPdef}
    \max_{y \in \R^m} \,\,
    \Big\{b^T y \,:\, \sum_{i=1}^m y_i A_i \preceq C\Big\}\,.
\end{equation}
By weak duality,
$\tr(CX) \geq b^T y$ holds for all pairs of feasible points $(X,y)$.
In particular, if the primal SDP is a feasibility problem $(C=0)$,
then
any solution of the dual SDP with positive objective value proves infeasibility of the primal SDP. %
We have that:

\begin{restatable}[]{observation}{obsSDPinfeas}
\label{obs:SDP_infeas}
The SDP in Eq.~\eqref{eq:SDP} is infeasible.
\end{restatable}

The proof is provided as an
exact rational
infeasibility certificate
obtained by rounding a numerical solution,
and is provided as a set of positive semidefinite matrices in Appendix~\ref{app:inf_cert}.
This yields the main result of this paper:
\begin{corollary}
The state $\eEight$ is entangled.
\end{corollary}
\begin{proof}
Combine Observations~\ref{obs:SDP} and~\ref{obs:SDP_infeas}.
\end{proof}

\noindent{\bf Lovász theta body. --- }
We now make connections to graph theory
and show that the entanglement of $\eEight$
can also be proven via a graph invariant.
The first four lines of the SDP in Eq.~\eqref{eq:SDP}
are also known in the context
of graph theory.
For a graph $G$ with edges denoted by $a\sim b$,
the Lovász theta number is defined as~\cite{lovasz1979shannon}\footnote{
Many other formulations of the Lovász number are known~\cite{
porumbel2022demystifyingcharacterizationsdpmatrices}.
},
\begin{align}\label{eq:lovasz}
\vartheta(G) = \quad \max \quad   &\sum_{a=1}^{N} M_{aa}\nn\\
\text{subject to} \quad
& M_{aa} = x_a
\quad \quad \forall a \in G\,, \nn\\
& M_{ab} = 0 \quad \quad\,\,\,\, \text{if}
\quad a\sim b\,, \nn\\
&\begin{pmatrix}
1 & x^T \\
x & M
\end{pmatrix}
\succeq 0\,.
\end{align}
The significance of the Lovász theta number comes from the fact that
it bounds the independence number $\alpha(G)$
and the chromatic number $\chi(\overbar{G})$~\cite{
Knuth_1994},
both quantities that are NP-hard to compute, as
\begin{equation}
\alpha(G) \leq \vartheta(G) \leq \chi(\overbar{G})\,,
\end{equation}
with $\overbar{G}$ the complement graph.
Also define the Lovász theta body as
\begin{align}
    \text{TH}(G) &=
    \Big\{
    \text{diag}(M) \,\Big|\,
    \begin{pmatrix}
        1 & x^T \\
        x & M
    \end{pmatrix} \succeq 0\,,\\
    &\quad\quad\,\,
    M_{aa} = x_a \,\,\forall a \in G\,,\,
    M_{ab} = 0 \,\,\text{ if }\, a \sim b
    \Big\}\,.\nn
\end{align}
Clearly
\begin{equation}
    \vartheta(G) = \max \big\{ \sum_{a \in G} \Delta_a \,|\, \Delta \in \textup{TH}(G)\big\}\,.
\end{equation}

Now let $\{E_a\}_{a=1}^m$ be a a set of Pauli strings.
Define its corresponding
anti-commutativity graph $G$
as the graph with
$m$ vertices,
where two vertices $a,b$ are connected,
written $a\sim b$,
if $E_a^{\dag} E_b = -E_b E_a^{\dag}$. Naturally, this encodes the third constraint of the SDP in Eq.~\eqref{eq:SDP} as a graph.

Consider the anti-commutativity graph
$G_{4+}$ of all Pauli strings $E_a \in \mathcal{E}_7$
of weight
{\it larger or equal than four},
where the weight of an element
$E_a$ is the number of its nontrivial tensor factors.
Therefore the entanglement of $\eEight$
can be understood as the fact that there is no point in
the convex body $\text{TH}(G_{4+})$
for which
$1 + \sum_{a\in G_{4+}} M_{aa} = 2^7$
holds.
\begin{observation}\label{obs:lovasz_infeasible}
   If $\eEight$ is separable, then
   there exists a point in $\textup{TH}(G_{4+})$
   with $1+ \sum_{a\in G_{4+}} M_{aa} = 2^7$.
\end{observation}
\begin{proof}
Suppose $\eEight$ is separable.
By Observation~\ref{obs:SDP}
the SDP in Eq.~\eqref{eq:SDP} is feasible.
The state $\eEight$ has vanishing expectation values
on all non-trivial Pauli operators of weight smaller or equal to three (c.f. Appendix \ref{app:correlations}),
\begin{equation}
 \tr\big((E_a \ot E_a^{\dag}) \eEight\big) = 0\,, \quad \text{if}
 \quad
 0 < \wt(E_a) < 4\,.
 \label{eq:0-diagonal-terms}
\end{equation}
Thus $\bGamma_{aa}=0$ must hold in
Eq.~\eqref{eq:SDP} for all $a$ with
$0 < \wt(E_a) < 4$.
As $\bGamma$ must be positive semidefinite,
so must all of its $2\times 2$ minors.
As a consequence, if a diagonal entry
$\bGamma_{aa}$ vanishes,
then so does the whole corresponding row and column,
i.e., $\bGamma_{ab} = \bGamma_{ba} = 0
$ for all $b= 0,\dots, N$.
As a consequence, we can truncate the matrix
$\bGamma$ without changing the objective value.
In terms of the anticommutativity graph of
$\mathcal{E}_7$,
this corresponds to
removing all vertices $a$
for which $0 < \wt(E_a) < 4$,
leading to the graph $G_{4+}$.
Furthermore by Eqs.~\eqref{eq:SwapExp} and~\eqref{eq:swap_expansion},
\begin{equation}
  \sum_{a=0}^N \tr\big((E_a \ot E_a^{\dagger})
  \eEight \big)
  = 2^7\,.
\end{equation}

We can thus relax the SDP in Eq.~\eqref{eq:SDP} to
\begin{align}\label{eq:SDP_Lovasz_feas}
\text{find} \quad   & M \nn\\
\text{subject to} \quad
& M_{aa} = x_a \quad \quad \forall a\in G_{4+}\,,\nn\\
& M_{ab} = 0 \quad \quad\,\,\,\, \text{if}
\quad a\sim b\,, \nn\\
&\begin{pmatrix}
1 & x^T \\
x & M
\end{pmatrix}
\succeq 0\,,\nn\\
&1 + \sum_{a\in G_{4+}} M_{aa} = 2^7\,.
\end{align}
If the SDP~\eqref{eq:SDP} is feasible, then so is
the SDP~\eqref{eq:SDP_Lovasz_feas}.
Let $M$ be a feasible point of the latter.
Then $\diag(M)$ is a point in $\textup{TH}(G_{4+})$
satisfying $1 + \sum_{a\in G_{4+}} M_{aa} = 2^7$.
This ends the proof.
\end{proof}
Similar to the SDP in Eq.~\eqref{eq:SDP},
this relaxed SDP of Eq.~\eqref{eq:SDP_Lovasz_feas}
is computationally expensive.
However, it is possible to run
a symmetry-reduced version, the dual of which has been been numerically solved in Ref.~\cite[Appendix B]{munne2024sdpboundsquantumcodes}.
This dual solution acts as an infeasible certificate for Eq.~\eqref{eq:SDP_Lovasz_feas},
showing that such point in
$\textup{TH}(G_{4+})$ cannot exist.
Thus we can also prove that $\eEight$ is entangled
via a relaxation involving the theta body.
Details and an exact certificate are given in Appendix~\ref{app:sym_red_Lovasz_app}.

\smallskip
\noindent
{\bf Entanglement witness. --- }
The open problem in Yu et al.~\cite{Yu2021}
explicitly also asks for an operational entanglement criterion.
We derive an operator witness;
i.e., an observable that, if negative on a state $\varrho$,
shows that $\varrho$ is entangled.
For this we relate our approach with that by Gois et al.~\cite{PhysRevA.107.062211}.

Consider two sets of observables $\{A_{a}\}_{a=1}^{m}$, $\{B_{a}\}_{a=1}^{m}$
with
anticommutativity graphs
$G_{A}$, $G_{B}$.
Then the following is an entanglement witness~\cite{PhysRevA.107.062211},
     \begin{equation}
     \label{eq:generalWitness}
    W = \sqrt{\vartheta(G_{A})\vartheta(G_{B})} \, \one \, - \sum_{a=1}^m A_{a}\otimes B_{a}\,.
    \end{equation}
That is,
$\tr(W\varrho^{\text{sep}})\geq 0$
for all states $\varrho^{\text{sep}}$ that are separable across the bipartition $A|B$;
while
$\tr(W\varrho^{\text{ent}})< 0$ holds
for at least one entangled state
$\varrho^{\text{ent}}$.
From
Eq.~\eqref{eq:generalWitness}
follows that
entanglement in $\eEight$ is
detected by the witness,
\begin{align}\label{eq:witness}
W_{\eEight} &=
\vartheta \big(G_{4+}\big) \one \,
- \sum_{a\in G_{4+}} \!\!E_a \ot E_a^{\dagger}\,.
\end{align}

However, obtaining $\vartheta(G_{4+})$
is computationally challenging due to adverse scaling.
In Appendix~\ref{app:sym_red_Lovasz_app} we compute a
 symmetry-reduced relaxation
with value
\begin{equation}
\label{eq:Lovasz_sym_value}    \vartheta_{\text{sym}}\big(G_{4+}\big)
    \approx 126.8876 \,.
\end{equation}
Due to the fact that $\vartheta_{\text{sym}}$ is a relaxation
 of $\vartheta$,
one has $\vartheta_{\text{sym}}(G_{4+})
 \geq \vartheta(G_{4+})$.
Surprisingly,
this relaxation is already enough to obtain a witness.
Indeed the operator
\begin{equation}
\label{eq:witness-sym}
W_{\eEight}' =
\vartheta_{\text{sym}} \big(G_{4+}\big) \one \,
- \sum_{a\in G_{4+}} E_a \ot E_a^{\dagger}\,,
\end{equation}
with
$\vartheta_{\text{sym}}\big(G_{4+}\big) \approx 126.8876$
also detects entanglement in $\eEight$,
as
$\tr(W' \eEight) < 0$ due to Eq.~\eqref{eq:SwapExp}
and the expansion in Eq.~\eqref{eq:swap_expansion}.
In fact, one can show (Appendix \ref{app:exact})
that $\vartheta_{\text{sym}}\big(G_{4+}\big) = \vartheta\big(G_{4+}\big)$ due to symmetries of the Pauli anti-commutativity graph, making the relaxation exact.

\smallskip
\noindent{\bf Conclusions. --- }
We have derived our main result,
$\eEight$ is entangled, via two semidefinite programming formulations.
Of these, the SDP~\eqref{eq:SDP}
is the strongest.
In turn,
the SDP~\eqref{eq:SDP_Lovasz_feas}
is a natural relaxation,
which has also already appeared in the
context of bounds on quantum codes~\cite{
munne2024sdpboundsquantumcodes}.
This last SDP then naturally links up to
a construction known from entanglement theory~\cite{PhysRevA.107.062211},
yielding a
computationally practical
witness [Eq.~\eqref{eq:witness}].
In contrast to previous works,
our constructions clarifies
the connection
to the symmetric extension hierarchy,
a well known computational method
for outer approximations to the set of separable states.
The key insight is to combine symmetric extension with moment matrices,
yielding an object that can more easily be truncated
than density matrices.

Several natural constraints
could be added to the SDP in Eq.~\eqref{obs:SDP},
for example constraints arising from positive partial transpose or state inversion conditions~\cite{PhysRevLett.111.030501, 0034-4885-67-3-R03}.
Also, purity type constraints from swap expectations,
e.g.,
$\tr\big((123) \sum_i p_i \mu^{\ot 3}\big)
=
2^{-7} \tr\big((12) \sum_i p_i \mu_i^{\ot 2}\big)$,
and structural relations, e.g., $\Gamma_{ab} \propto \Gamma_{\a\b}$ if
$\av{E_a}\av{E_b^\dag}\av{E_a^\dag E_b} \propto
\av{E_\a}\av{E_\b^\dag}\av{E_\a^\dag E_\b}$,
can be imposed on the moment matrix, c.f.~\cite[Eq.~(68)]{munne2024sdpboundsquantumcodes}.
While not needed to show entanglement in
$\eEight$, these can strengthen our method.
Finally, one could generalize the approach to
detect entanglement in
non-symmetric states.

Our approach
establishes remarkable
links between
entanglement detection, graph theory, and coding bounds that revolve around the Lovász theta number.
One can hope that the Terwilliger methods develop here can be used to determine analytical bounds on the (anti-) commutation index (also known as $\beta$-number)~\cite{PRXQuantum.5.020318, PhysRevLett.132.200202},
relevant in the study of quantum Hamiltonians~\cite{hastings2023optimizing,cbqf-d24r}
and
stabilizer testing~\cite{10.1145/3717823.3718201}.

\smallskip
{\it Acknowledgments.}
We thank
Otfried Gühne,
Chau Nguyen,
Siyuan Qi,
Albert Rico,
Lisa Weinbrenner,
Zhen-Peng Xu,
and
Xiao-Dong Yu,
for fruitful discussions.
JS, SL, FH were funded in whole or in part by the National Science Centre, Poland 2024/54/E/ST2/00451
and by the Polish National Agency for Academic Exchange under the Strategic Partnership Programme grant BNI/PST/2023/1/00013/U/00001.
For the purpose of Open Access,
the authors have applied a CC-BY public copyright licence to any Author Accepted Manuscript (AAM) version arising from this submission.
GAM was supported by NCN grant no. 2024/53/B/ST2/04103.

\section{Rains' Lemma}
\label{app:aux}

\RainsLemma*

\begin{proof}
We rewrite,
\begin{align}
    \big[(12) (A \ot A^\dag)\big]^{T_1}
    &=
    \big[(12) (A \ot I) (I \ot A^\dag)\big]^{T_1}
    \nn\\
    &=
    \big[(12) (A \ot I)\big]^{T_1} (I \ot A^\dag)
    \nn\\
    &=
    [(I \ot A) (12) \big]^{T_1} (I \ot A^\dag)
    \nn\\
    &=
    (I \ot A) (12)^{T_1} (I \ot A^\dag) \succeq 0\,,
\end{align}
where in last inequality we
used the fact that that
$(12)^{T_1} = d \dyad{\phi^+} \succeq 0$ is proportional to the Bell state with local dimensions $d$,
and $XAX^\dag \succeq 0 $ for all $X$ and all $B\succeq 0$.
This ends the proof.
\end{proof}

\section{Properties of $\Gamma$}
\label{app:gamma}
\subsection{Auxilliary proof for Observation~\ref{obs:SDP}}
We now derive Eqs.~\eqref{eq:Gamma_pos} and $\eqref{eq:Gamma_entries}$.
Recall that $\Gamma$ has entries
$    \Gamma _{ab} = \tr(R_{a}^{\dag}R_{b}\Psi)$,
where
\begin{align}
    \Psi &= [(13)\eEight^{(3)}]^{T_{1}} \,,\nn\\
    R_{a} &= \frac{1}{\sqrt{2}}(13)^{T_{1}}(I \ot E_{a} \ot E_{a}^{\dag})\,, \quad E_{a} \in \mathcal{E}_{7}\,,
\end{align}
with $\mathcal{E}_{7}$
the set of all Pauli strings of length $7$.

To show Eq.~\eqref{eq:Gamma_pos}, note that for all $x \in \C^{4^n}$,
\begin{align}
   \nn x^{\dag}\Gamma x  & =\sum_{ab}\bar{x}_{a}\Gamma_{ab}x_{b} \\
 \nn   & =\sum_{ab} \bar{x}_{a}\tr(R_{a}^{\dag}R_{b}\Psi)x_{b}\\
   \nn & =\tr(\sum_{b}x_{b}R_{b}\Psi \sum_{a}\bar{x}_{a}R_{a}^{\dag}) \\
   & =\tr(vv^{\dag}\Psi)\geq 0\,,
\end{align}
where
$v = \sum_a \bar{x}_a R_a^\dag$. Since $\Psi$ is positive semidefinite,  the last inequality follows from the fact that the trace inner product of two semidefinite operators is non-negative.

To show Eq.~\eqref{eq:Gamma_entries}
we rewrite the expression
for the entries of the moment matrix as
\begin{align}
     \Gamma_{ab} =& \tr\big(R_{a}^{\dag}R_{b}[(13)\eEight^{(3)}]^{T_{1}}\big) \nn\\=& \tr\big((R_{a}^{\dag}R_{b})^{T_{1}}(13)\eEight^{(3)}\big)\,,
\end{align}
where we have used that
$\tr(A^{T_{1}}B) = \tr(AB^{T_{1}})$
holds for all bipartite operators $A$, $B$.

We then evaluate,
\begin{align}\label{eq:eval_R}
    &(R_{a}^{\dag}R_{b})^{T_{1}} \nn\\
    &=
    \frac{1}{2}\Big( \big[(13)^{T_{1}}
    (I \otimes E_{a} \otimes E_{a}^{\dag})\big]^{\dag}(13)^{T_{1}}
    (I \otimes E_{b} \otimes E_{b}^{\dag})
    \Big)^{T_{1}} \nn\\
    &=
    \frac{1}{2}\Big(
    (I \otimes E_{a}^{\dag}\otimes E_{a})  \big[(13)^{T_{1}}\big]^{\dag}
    (13)^{T_{1}}
    (I \otimes E_{b} \otimes E_{b}^{\dag})
    \Big)^{T_{1}}
    \nn\\
    &=\big((I \otimes E_{a}^{\dag}\otimes E_{a})  (13)^{T_{1}}(I \otimes E_{b} \otimes E_{b}^{\dag})\big)^{T_{1}}
    \nn\\
    & = (I \otimes E_{a}^{\dag}\otimes E_{a})  (13)(I \otimes E_{b} \otimes E_{b}^{\dag}) \nn\\
    & = (I \otimes E_{a}^{\dag}\otimes E_{a})  (E_{b}^{\dag} \otimes E_{b} \otimes I)(13)\,,
\end{align}
where the third line follows from the fact that $(13)^{T_{1}} = 2 \dyad{\phi^{+}}_{13}
\ot \one_2$
with $\dyad{\phi^{+}}$
projector onto the Bell state
$\ket{\phi^{+}} = (\ket{00} + \ket{11})/\sqrt{2}$ .
For the fourth and fifth line of Eq.~\eqref{eq:eval_R}
we used the identities,
\begin{align}
     (X^{T_{1}})^{T_{1}} &= X\,, \nn\\
     \big(X(I\ot Y)\big)^{T_{1}} &= X^{T_{1}}(I\ot Y)\,,\nn\\
    (13)(A\ot B \ot C) &= (C\ot B \ot A)(13)\,.
\end{align}

By Eq.~\eqref{eq:eval_R}, the entries of the moment matrix are then
\begin{align}
\label{eq:momentEntries}
&\Gamma_{ab} = \tr(R_{a}^{\dag}R_{b} \Psi) \nn\\
&=\tr\big(
\big[(I \otimes E_{a}^{\dag}\otimes E_{a})(E_{b}^{\dag} \otimes E_{b}\otimes I)(13)\big]^{T_1}\big[(13)\eEight^{(3)}\big]^{T_1}\big) \nn\\
  \nonumber
  & =\tr\big((I \otimes E_{a}^{\dag}\otimes E_{a})(E_{b}^{\dag} \otimes E_{b}\otimes I)\eEight^{(3)}\big) \\
\nonumber
& = \tr\big(
(I\ot E_{a}^{\dag}\ot E_{a})(E_{b}^{\dag}\ot E_{b}\ot I)(\sum_{i}p_{i}\mu_{i}^{\ot 3})
\big)\\
&=\sum_{i}p_{i}\tr(E_{b}^{\dag}\mu_{i})\tr(E_{a}\mu_{i})\tr(E_{a}^{\dag}E_{b}\mu_{i})\,,
\end{align}
where we have used  $(13)^{2} = I$ and
the decomposition $\eEight^{(3)} = \sum_{i}p_{i}\mu_{i}^{\ot 3}$.

\subsection{Moment matrix formulation}
\label{app:MomentMatrixFormulation}

To obtain a better intuition, we formulate $\Gamma$ as a positive-semidefinite Gram matrix.

 Let us use the fact that the $\mu_{i}$ state is a pure state and it can be written as $\dyad{\mu_{i}}$, for some vector $\ket{\mu_{i}}$. We define the set of $N+1 = 4^7$ vectors\
\begin{equation}
\label{eq:explicitLambda2}
    \ket{E_{\beta}^{(i)}} = \tr(E_{\beta}^{\dag}\mu_{i})\big(E_{\beta}\ket{\mu_{i}}\big)\,,\quad \beta =0,\ldots, N \,.
\end{equation}
 Now, it follows that
\begin{align}
\label{eq:convex_sum2}
\nonumber\Gamma_{ab} &=\sum_{i}p_{i} \tr(E_{b}^{\dag}\mu_{i})\tr(E_{a}\mu_{i})\tr(E_{a}^{\dag}E_{b}\mu_{i}) \\&=\sum_{i}p_{i}\braket{E_{a}^{(i)}}{E_{b}^{(i)}}\,.
\end{align}
The full moment matrix then reads as
\begin{equation}
   \Gamma =  \sum_{i}p_{i}
\begin{pmatrix}
 \text{--- } & \bra{E_0^{(i)}}&  \text{--- }\\
    \\
    \vdots & \vdots & \vdots\\
    \\
     \text{--- } & \bra{E_{N-1}^{(i)}}^{} &  \text{--- }\\
\end{pmatrix}
\begin{pmatrix}
|&\ldots & |\\
& & \\
     \ket{E_0^{(i)}}& \ldots & \ket{E_{N-1}^{(i)}}\\
       && \\
   |  &\ldots & |\\
\end{pmatrix}
\,,
\end{equation}
which means that the $\Gamma$ is a convex combination of Gram matrices composed of vectors $\ket{E_{\beta}^{(i)}}$.

\subsection{Diagonal elements}
\label{app:correlations}

Here, we derive coarse-grained constraints
on the diagonal entries of the moment matrix. From these, we then obtain the exact diagonal terms of the moment matrix that appear in Eq.~\eqref{eq:SDP} of the main text.

For any $i\in \{1,\ldots,n\}$, the sum over diagonal entries corresponding to Pauli strings of weight $i$ takes the form
\begin{align}
   \nonumber& \sum\limits_{\substack{E_{a}\in\mathcal{E}_n\\
   \text{wt}(E_{a})= i }}\Gamma_{aa}=\sum\limits_{\substack{E \in \mathcal{E}_n\\
   \text{wt}(E)= i }}\tr\big[(E\otimes E^{\dag})\eEight\big] \\&= \sum\limits_{\substack{S \subseteq [n]\\ |S|=i}}\ \sum\limits_{\substack{E\in \mathcal{E}_n\\\text{supp}(E) = S}}
    \tr\big[(E\otimes E^{\dag})\eEight\big]\label{eq:app-coarse-grained}
\end{align}
where $[n]:=\{1,\ldots,n\}$.
To evaluate this sum, we relate these terms to the SWAP operator. This operator acting on a subsystem $S\subseteq [n]$ can be expanded in the Pauli basis as the sum over all strings whose support is contained in $S$,
\begin{equation}
    \pi_{S} = \bigotimes_{k\in S} \text{SWAP}_{A_{k},B_{k}} = \frac{1}{2^{|S|}}\sum\limits_{\substack{E\in\mathcal{E}_n\\\text{supp}(E)\subseteq S}} E\otimes E^{\dag} \,.
\end{equation}

Note that the constraint in Eq.~\eqref{eq:app-coarse-grained} requires summing over strings with support equal to $S$, whereas the SWAP formula yields a sum over all its subsets.
When summing over all subsets $T\subseteq S$ with an alternating sign factor, that depends on the size difference between $T$ and $S$, we effectively filter out the contributions of all subsets different from $S$.
This alternating sum reformulates Eq.~\eqref{eq:app-coarse-grained} as
\begin{align}
\label{eq:constraintSum}
   \nonumber& \sum\limits_{\substack{E_{a}\in\mathcal{E}_n \\\text{wt}(E_{a})= i }} \Gamma_{aa}\\ \nonumber
   =&\sum\limits_{\substack{S\subseteq [n]\\ |S|=i}}\sum\limits_{\substack{T\subseteq S}} \sum\limits_{\substack{E\in\mathcal{E}_n\\\text{supp}(E)\subseteq T}} (-1)^{|S|-|T|}\tr(E\ot E^{\dag}\eEight)
   \\ =&\sum\limits_{\substack{S\subseteq [n]\\|S|=i}}\sum\limits_{T\subseteq S} (-1)^{|S|-|T|}2^{|T|}\tr(\pi_{T}\eEight) \,.
\end{align}

Finally, we simplify this expression by evaluating the SWAP expectation value defined in Eq.~\eqref{eq:SwapExp} and counting the subsets by their size as
\begin{align}
    \label{eq:exactVal}
      \sum\limits_{\substack{E_{a}\in\mathcal{E}_n \\\text{wt}(E_{a})= i }} \Gamma_{aa}
   = \sum_{j=0}^{i} \binom{n}{i}\binom{i}{j}\frac{(-1)^{i-j}2^{j}}{\min\{2^{j},2^{n-j}\}} \,,
\end{align}
where the sum runs over the size of the subset $|T|=j$. The first binomial coefficient counts all possible choices for the set $S$ of size $|S|=i$ in the $n$-element string, while the second accounts for the elements inside $S$ to form the set $T$. For completeness, we present the numerical value of the sum in Eq.~\eqref{eq:exactVal} as the vector
\begin{equation}
\label{eq:app-vector-DiagA}
    A=[1,0,0,0,35,42,28,22]\,,
\end{equation}
where the $i$-th entry corresponds
to weight $i$.

Additionally, the individual diagonal entries $\bar\Gamma_{aa}$ of the moment matrix can be derived directly from Eq.~\eqref{eq:exactVal}. This relies on two key symmetries of $\Phi_{\text{E}8}$.

First, it presents a strong permutational invariance. For any permutation $\sigma\in S_7$
acting jointly on both $A$ and $B$ parties, the state is invariant:
\begin{equation}
    (\sigma_A\ot\sigma_B)\eEight(\sigma_A\ot\sigma_B)^\dagger=\eEight\,.
    \label{app:eq-state-sigma-invariance}
\end{equation}

Second, the state is composed of tensor products of symmetric and antisymmetric projectors, $P_\pm$, in the subsystems $(A_i,B_i)$. These projectors are invariant under the joint action of any unitary $U\in\mathrm{SU}(2)$:
\begin{equation}
    \big(U\otimes U\big)\,P_\pm\, \big(U^\dagger\otimes U^\dagger\big)=P_\pm\,.
    \label{eq:app-invariance-p-pm}
\end{equation}
This symmetry extends to the state $\Phi_{\text{E}8}$,
\begin{equation}
    \bigotimes_{i}\Big(U_{A_i}^{(i)}\otimes U_{B_i}^{(i)}\Big)\,\eEight\, \bigotimes_{i}\Big({U_{A_i}^{(i)}}\otimes {U_{B_i}^{(i)}}\Big)^\dagger=\eEight\,,\label{app:eq-state-unitary-invariance}
\end{equation}
where possibly different unitaries $\{U^{(i)}\}_{i=1}^7$ act on pairs $\{{(A_i,B_i)}\}_{i=1}^7$.

These two properties allow us to conclude that the diagonal elements of the moment matrix satisfy
\begin{equation}
    \Gamma_{aa}=\tr\big[(E_a\otimes E_a^\dagger)\Phi_{\text{E}8}\big]=\tr\big[(E_b\otimes E_b^\dagger)\Phi_{\text{E}8}\big]=\Gamma_{bb}\,,
\end{equation}
if $\text{wt}(E_a)=\text{wt}(E_b)$, as any Pauli string can be transformed into another one with the same weight by the transformations defined in Eqs.~\eqref{app:eq-state-sigma-invariance} and~\eqref{app:eq-state-unitary-invariance}. Consequently, the explicit expression of the diagonal entries reads
\begin{equation}
\label{eq:DiagonalTerm}
    \Gamma_{aa}=\dfrac{1}{3^i}\sum_{j=0}^{i}\binom{i}{j}\frac{(-1)^{i-j}2^{j}}{\min\{2^{j},2^{n-j}\}} \,,
\end{equation}
which yields the vector in Eq.~\eqref{eq:app-vector-DiagA}, divided by the total number $3^i\binom{n}{i}$ of Pauli strings with weight $i$ for each of the entries:
\begin{equation}
\label{eq:EntriesVector}
    a=\bigg[1,0,0,0,\frac{1}{3^4},\frac{2}{3^5},\frac{4}{3^6},\frac{22}{3^7}\bigg]\,.
\end{equation}

\section{Infeasibility certificate
for Observation 1}
\label{app:inf_cert}

\begin{remark}
For the derivations that follow,
we for simplicity
do not impose certain constraints that
appear in the related work on coding bounds Ref.~\cite{munne2024sdpboundsquantumcodes}.
See Appendix~\ref{app:relation} for details.
 \end{remark}

\obsSDPinfeas*
\begin{proof}

To show this requires the following steps:
\begin{enumerate}
    \item[A.] Derive a relaxation of Eq.~\eqref{eq:SDP}.
    \item[B.] Block-diagonize this relaxation.
    \item[C.] Derive its dual semidefinite program.
    \item[D.]
     Solve this dual SDP exactly using a semidefinite programming solver with heuristic rounding to rational numbers.
\end{enumerate}
We proceed with these steps below.
Most of the following theory holds for general $n$, and we will set $n=7$ later.

\subsection{Relaxation}

Consider the set of n-qubit Pauli strings,
and consider
wreath product $S_3 \wr S_n$~\cite{Ceccherini-Silberstein_Scarabotti_Tolli_2014}
that:
\begin{enumerate}
    \item permutes the tensor factors by $S_n$, and
    \item independently permutes the non-identity Pauli matrices at each coordinate by $S_3$.
\end{enumerate}

We now relax the SDP of
Observation~\ref{obs:SDP}
by averaging over the group
$S_3 \wr S_n$.
In particular, if
$\bGamma$ is feasible for SDP~\eqref{eq:SDP},
then the averaged matrix
$\tilde \Gamma$
(an element of the Terwilliger algebra~\cite{GIJSWIJT20061719})
is feasible for the averaged SDP.
As a consequence, if the averaged SDP is infeasible, so is the original SDP~\eqref{eq:SDP}.

Denote by $\supp(E_a)$ the support of a Pauli string $E_a$, that is, the subsystems it acts non-trivially on,
and let the weight $\wt(E_a) = |\supp(E_a)|$ be its size.
Let a permutation
$\pi$
(which will choose to be in $S_3 \wr S_n$) act on $\Gamma$ by permuting its rows and columns as,

\begin{equation}
\label{eq:transformation}
    \pi(\bar\Gamma)_{ab} = \bar{\Gamma}_{\pi^{-1}(a) \pi^{-1}(b)}\,.
\end{equation}
That is,
$\pi$ acts on $\bGamma$ by adjoint unitary action.
With some abuse of notation we will also write
$\pi(\bar\Gamma) = \pi\bar\Gamma \pi^{-1
}$.

Let $(E_a)_{a=0}^{N}$ with $N=4^{n-1}$ be vector of Pauli strings that index~$\Gamma$.
Given row $a$ and column $b$, consider
\begin{align}\label{eq:invariants}
\wt(E_a)&=i \,, \nn\\
\wt(E_b)&=j\,,\nn\\
|\supp(E^\dagger_a)\cap \supp(E_b)\rvert&= t \,,\nn\\
\wt(E^\dagger_aE_b) \, &= \, i+j-t-p \,.
\end{align}

Denote by $\Pi$ the set of all permutations that, for every row labeled by $E_a$ and column labeled by $E_b$,
leave tuple $(i,j,t,p)$
from Eq.~\eqref{eq:invariants} invariant.
We then define the \textit{symmetrized} version of our moment matrix as the average of the original matrix over elements of $\PP$,
\begin{equation}
    \tilde \Gamma = \frac{1}{|\PP|}
    \sum_{\pi \in \Pi} \pi (\bGamma )\,.
    \label{eq:app-gamma-twirled}
\end{equation}
As a consequence of this average, the entries of the symmetrized moment matrix depend only on the tuple $(i,j,t,p)$.
In particular, for every pair $E_a,E_b$
with parameters $(i,j,t,p)$  one has
\begin{align}\label{eq:lambda_anycode}
    \tilde\Gamma_{ab}=
    x^{t,p}_{i,j}= \frac{1}{\gamma^{t,p}_{i,j}}\!\!\!\! \sum\limits_{
        \substack{E_a,E_b\in \mathcal{E}_n \\|s(E^\dagger_a)| \,=\, i
            \\ |s(E_b)| \,= \,j
            \\ \lvert s(E^\dagger_a)\cap s(E_b)\rvert \,= \,t
             \\ \lvert s(E^\dagger_aE_b)\rvert \, = \, i+j-t-p \\}} \!\!\!\!\!\!{\overbar\Gamma}_{ab}  \,.
\end{align}
Here $\gamma^{t,p}_{i,j}$ is a normalization factor
\begin{align}
\gamma_{i,j}^{t,p}
	&= 3^{i+j-t}2^{t-p}\binom{n}{p,t-p,i-t,j-t}\,,
\end{align}
and the multinomial coefficient is defined as
\begin{align}
\binom{n}{a_1,\dots,a_r}= \frac{n!}{a_1!\dots a_r!(n-\sum^r_{\ell=1} a_\ell)!}\,.
\end{align}
Since $\pi$  acts on $\bGamma$ by unitary adjoint action [Eq.~\eqref{eq:transformation}],
it follows that if
$\bGamma \succeq 0$,
then $\tilde \Gamma \succeq 0$.

\subsection{ Block diagonalization}
Using the framework introduced in the previous section, we recall from Ref.~\cite{GIJSWIJT20061719} that $\tilde \Gamma \succeq 0$ if and only if
\begin{align}
\label{eq:symmetrizePSD}
\bigoplus_{\substack{a,k \in \N_0\\ 0\leq a\leq k\leq n+a-k}} \left(\sum\limits_{\substack{t,p\in \N_0 \\ 0 \leq p \leq t \leq i,j \\ i+j\leq t+n}}\alpha(i,j,t,p,a,k)x_{i,j}^{t,p}\right)_{i,j=k}^{n+a-k}   \!\!\!\!\!\!\!\!\succeq 0\,.
\end{align}
Here $\alpha(i,j,t,p,a,k)$ are combinatorial factors,
\begin{align}\label{eq:gamma}
	\alpha(i,j,t,p,a,k) &= 3^{\frac{1}{2}\left(i+j\right)-t}\beta_{i-a,j-a,k-a}^{n-a,t-a} \\&\times\sum\limits_{g=0}^{p}2^{t-a-p+g}\left(-1\right)^{a-g}\binom{a}{g}\binom{t-a}{p-g}\,, \nn
\end{align}
with
\begin{align}
	\beta_{i,j,k}^{m,t}&=\sum\limits_{u=0}^{m}\left(-1\right)^{t-u}
    \binom{u}{t}\binom{m-2k}{m-k-u} \nn\\&\quad \quad\quad\times\binom{m-k-u}{i-u}\binom{m-k-u}{j-u}\,.
\end{align}

Following Ref.~\cite{munne2024sdpboundsquantumcodes}
we relax the SDP of Eq.~\eqref{eq:SDP} to:

\bigskip
\noindent\underline{primal:}
\begin{align}\label{eq:sdpx_relax}
	\textnormal{find} 		\quad & x^{t,p}_{i,j}  \nn \\
	\textnormal{s.t.} 	\quad
	&  x^{0,0}_{0,0}=1 \\
    & x^{0,0}_{i,0}=x^{i,i}_{i,i}, \quad \text{for} \quad i=1,\dots, n \nn \\
    &x^{t,p}_{i,j}=0 \quad\quad \text{if}\quad  t-p \quad  \text{is odd}\,, \nn \\
    &\gamma^{i,i}_{i,i}x^{i,i}_{i,i}=\sum_{\substack{E\in \mathcal{E}_n \\\wt(E)=i}}\!\!\tr\big( (E\otimes E^\dagger)
   \eEight \big) \,, \quad i=1,\dots, n \nn \\
	&\bigoplus_{\substack{a,k \in \N_0\\ 0\leq a\leq k\leq n+a-k}} \left(\sum\limits_{\substack{t,p\in \N_0 \\ 0 \leq p \leq t \leq i,j \\ i+j\leq t+n}}\alpha(i,j,t,p,a,k)x_{i,j}^{t,p}\right)_{i,j=k}^{n+a-k}   \!\!\!\!\!\!\!\!\succeq 0\,, \nn
\end{align}
The SDP in Eq.~\eqref{eq:sdpx_relax}
is an adaptation of Ref.~\cite[Section~9.1]{munne2024sdpboundsquantumcodes}:
first, without the off-diagonal constraints $x_{i,j}^{t,p} = 0$ if $i+j-t-p < \delta$;
and second, fixing every $x_{i,i}^{i,i}\gamma_{i,i}^{i,i}$, instead of fixing their sum.
Further details on why this change happens are given in Appendix~\ref{app:relation}.

\subsection{Dual SDP}\label{sec:dualsdp}
Let us now derive the SDP
that is dual to  Eq.~\eqref{eq:sdpx_relax}.
It reads:

\bigskip
\noindent\underline{dual:}
\begin{align}\label{eq:dual_sol}
	\beta^* =
    \max_{Y^{(a,k)}} \quad
	&   -y_{0,0}^{0,0} - \sum^n_{i=1}  (2y_{i,0}^{0,0} + y^{i,i}_{i,i} ) \\ \nn &\quad\quad\quad\quad \times \sum_{\substack{E\in \mathcal{E}_n \\\wt(E)=i}}\!\! \tr\big( (E \otimes E^\dagger)
   \eEight \big)  \nn \\
	\text{s.t.}  \quad & Y^{(a,k)} \succeq 0\,, \nn \\
	&  y^{t,p}_{i,j}=0
	\quad  \text{if} \quad i,j\neq 0\, \,\, \text{and}  \,\,\, t-p \,\, \text{is even } \nn
\end{align}
where $Y^{(a,k)}$  are real matrices with indices constrained by $0 \leq a \leq k \leq n + a -k $ and
\begin{align}
    y^{t,p}_{i,j} = \frac{1}{\gamma^{t,p}_{i,j}}\sum^{\min(i,j)}_{k=0} \sum^k_{\substack{a=\max(i,j)\\\,+k-n}}  \alpha(i,j,t,p,a,k) Y^{(a,k)}_{i-k,j-k}.
\end{align}

    This dual is a variation of~\cite[Proposition 26]{munne2024sdpboundsquantumcodes}. Due to SDP duality, the dual variable $w$ from that SDP corresponds to the primal constraint $(2^n-1)=\sum^n_{i=1} x^{i,i}_{i,i}\gamma^{i,i}_{i,i}$.
    This constraint does not appear in the SDP of Eq.~\eqref{eq:sdpx_relax}, instead we have as primal constraints \begin{equation}
        \sum_{{E\in \mathcal{E}_n:\,\wt(E)=i}} \tr\big( (E \otimes E^\dagger)
   \eEight \big) = x^{i,i}_{i,i}\gamma^{i,i}_{i,i}\,,
    \end{equation}
    for $i=1,\dots, n$.
   These correspond to dual variables $w_i$ which are introduced by modifying~\cite[Proposition 26]{munne2024sdpboundsquantumcodes}:
    \begin{itemize}

         \item[i)] The dual variable  is changed:
        $w\rightarrow w_i$.
        \item[ii)] In the objective function, change
        \begin{align}
            (2^n-1)w \,\,\rightarrow\,\, \sum_{\substack{E\in \mathcal{E}_n\\\wt(E)=i}}\tr\big( (E \otimes E^\dagger)
   \eEight \big)w_i \,.
        \end{align}

\item[iii)] Change the constraint
\begin{align}
    y^{t,p}_{i,j}=0 \quad&\text{if}\quad
    i,j \neq 0
    \quad \text{and}\quad
    t-p \text{ even}
    \nn\\
    &\quad\quad\quad  \text{and}\quad
    i+j-t-p \geq \delta\,,
\intertext{to the constraint}
    y^{t,p}_{i,j}=0 \quad&\text{if}\quad
    i,j \neq 0
    \quad \text{and}\quad
    t-p \text{ even}\,.
\end{align}

\item[iv)] Defining the slack variables $w_i$ in terms of $y^{i,i}_{i,i}$ and $y^{0,0}_{i,0}$ leads to objective function of  Eq.~\eqref{eq:dual_sol}.

\end{itemize}

\subsection{4. Solve dual SDP and round}
The primal SDP~\eqref{eq:sdpx_relax} is a feasibility problem, and thus $C = 0$ in Eq.~\eqref{eq:PrimalSDPdef}. As a consequence, a feasible point of the dual SDP
\eqref{eq:dual_sol}
with positive objective value signals
infeasibility of the primal SDP.

The matrices listed below constitute  a feasible point
of SDP~\eqref{eq:dual_sol}
with positive objective value $\beta = 1$,

\begin{widetext}
\begin{center}
\scriptsize
\renewcommand{\arraystretch}{1.25}      
\setlength{\arraycolsep}{2pt}
\begin{align}
\nn Y^{(0,0)} = &\begin{bmatrix}
6744z + 2188885 & 0 & 0 & 0 & 104766z - 331211 & -72585z - 163934 & 27937z - 537566 & -12674z - 680320 \\
0 & -2z + 767759 & 0 & 0 & 0 & 0 & 0 & 0 \\
0 & 0 & -z + 54083 & 0 & 0 & 0 & 0 & 0 \\
0 & 0 & 0 & -z + 18880 & 0 & 0 & 0 & 0 \\
104766z - 331211 & 0 & 0 & 0 & -z + 10949 & 18779 & 33614 & 46156 \\
-72585z - 163934 & 0 & 0 & 0 & 18779 & -z + 41486 & 61654 & \frac{3043753}{35} \\
27937z - 537566 & 0 & 0 & 0 & 33614 & 61654 & -z + 116898 & 151157 \\
-12674z - 680320 & 0 & 0 & 0 & 46156 & \frac{3043753}{35} & 151157 & -17z + 280007
\end{bmatrix}, \allowdisplaybreaks\\
\nn Y^{(0,1)} = &\begin{bmatrix}
-84z + 32245878 & 0 & 0 & 0 & 0 & 0 \\
0 & -\frac{367}{30}z + \frac{32590321}{30} & 0 & 0 & 0 & 0 \\
0 & 0 & 2z + 140265 & 0 & 0 & 0 \\
0 & 0 & 0 & 13586 & 15723 & -39534 \\
0 & 0 & 0 & 15723 & -z + 21103 & -29608 \\
0 & 0 & 0 & -39534 & -29608 & -4z + 229721
\end{bmatrix}, \allowdisplaybreaks\\
 \nn Y^{(0,2)} =& \begin{bmatrix}
-\frac{104}{3}z + \frac{31108352}{3} & 0 & 0 & 0 \\
0 & \frac{247}{18}z + \frac{72477905}{72} & 0 & 0 \\
0 & 0 & -21z + 190334 & -\frac{6600763}{72} \\
0 & 0 & -\frac{6600763}{72} & 36z + 473063
\end{bmatrix}, \allowdisplaybreaks\\
\nn Y^{(0,3)} =& \begin{bmatrix}
-\frac{353}{3}z + \frac{53700905}{12} & 0 \\
0 & 107z + \frac{133285}{2}
\end{bmatrix}, \allowdisplaybreaks\\
\nn Y^{(1,1)} =& \begin{bmatrix}
10z + 24263194 & 0 & 0 & 0 & 0 & 0 & 0 \\
0 & 3z + 957137 & 0 & 0 & 0 & 0 & 0 \\
0 & 0 & 2z + 213736 & 0 & 0 & 0 & 0 \\
0 & 0 & 0 & -2z + 18577 & -4392 & -24417 & \frac{129157}{5} \\
0 & 0 & 0 & -4392 & z + 5024 & 13800 & -2541 \\
0 & 0 & 0 & -24417 & 13800 & -2z + 65599 & 26742 \\
0 & 0 & 0 & \frac{129157}{5} & -2541 & 26742 & -27z + 213536
\end{bmatrix}, \allowdisplaybreaks\\
\nn Y^{(1,2)} =& \begin{bmatrix}
-\frac{1405}{6}z + \frac{213766795}{6} & 0 & 0 & 0 & 0 \\
0 & -\frac{883}{8}z + \frac{173414105}{32} & 0 & 0 & 0 \\
0 & 0 & -33z + 218776 & -\frac{22465829}{144} & \frac{1644555}{4} \\
0 & 0 & -\frac{22465829}{144} & -10z + 763573 & -\frac{544565}{4} \\
0 & 0 & \frac{1644555}{4} & -\frac{544565}{4} & 5z + 1031762
\end{bmatrix}, \allowdisplaybreaks\\
\nn Y^{(1,3)} =& \begin{bmatrix}
-338z + \frac{134968445}{4} & 0 & 0 \\
0 & \frac{291}{4}z + \frac{1792743}{4} & \frac{38117659}{48} \\
0 & \frac{38117659}{48} & -\frac{39}{8}z + \frac{11423859}{8}
\end{bmatrix},\allowdisplaybreaks\\
\nn Y^{(1,4)} =& \begin{bmatrix}
-\frac{139}{2}z + 2155897
\end{bmatrix}, \allowdisplaybreaks\\
\nn Y^{(2,2)} =& \begin{bmatrix}
35z + 14993587 & 0 & 0 & 0 & 0 & 0 \\
0 & -21z + 1738517 & 0 & 0 & 0 & 0 \\
0 & 0 & -5z + 109328 & 46435 & 128388 & -\frac{824124}{5} \\
0 & 0 & 46435 & 2z + 244491 & -224627 & \frac{699903}{5} \\
0 & 0 & 128388 & -224627 & -2z + 652767 & -103783 \\
0 & 0 & -\frac{824124}{5} & \frac{699903}{5} & -103783 & -35z + 1284622
\end{bmatrix},
\allowdisplaybreaks\\
\nn Y^{(2,3)} =& \begin{bmatrix}
-354z + \frac{75611175}{2} & 0 & 0 & 0 \\
0 & -46z + \frac{8456720}{3} & \frac{4881140}{27} & -\frac{9400155}{4} \\
0 & \frac{4881140}{27} & 43z + 234163 & \frac{756605}{2} \\
0 & -\frac{9400155}{4} & \frac{756605}{2} & -25z + 3267902
\end{bmatrix},
\allowdisplaybreaks\\
\nn Y^{(2,4)} =& \begin{bmatrix}
114z + 3679177 & -\frac{227366551}{36} \\
-\frac{227366551}{36} & -\frac{2169}{16}z + \frac{173699349}{16}
\end{bmatrix}, \allowdisplaybreaks\\
\nn Y^{(3,3)} =& \begin{bmatrix}
-93z + 8264420 & 0 & 0 & 0 & 0 \\
0 & 195z + 491796 & -485093 & -\frac{2641175}{32} & 1219792 \\
0 & -485093 & -14z + 938299 & 617405 & -770467 \\
0 & -\frac{2641175}{32} & 617405 & -8z + 648307 & \frac{2412575}{8} \\
0 & 1219792 & -770467 & \frac{2412575}{8} & 16z + 3446498
\end{bmatrix},
\allowdisplaybreaks\\
\nn Y^{(3,4)} =& \begin{bmatrix}
42z + 18555316 & \frac{82089313}{12} & \frac{83657415}{8} \\
\frac{82089313}{12} & \frac{181}{16}z + \frac{191223671}{16} & -\frac{27416555}{4} \\
\frac{83657415}{8} & -\frac{27416555}{4} & -\frac{187}{2}z + \frac{36150719}{2}
\end{bmatrix},\allowdisplaybreaks\\
\nn Y^{(3,5)} =& \begin{bmatrix}
\frac{6225}{16}z + \frac{109987}{16}
\end{bmatrix},
\allowdisplaybreaks\\
\nn Y^{(4,4)} =& \begin{bmatrix}
500z + 7512496 & -\frac{5486233}{3} & -\frac{6936135}{2} & -5363008 \\
-\frac{5486233}{3} & \frac{4331}{24}z + \frac{8316451}{8} & -\frac{2136405}{4} & 3091740 \\
-\frac{6936135}{2} & -\frac{2136405}{4} & -\frac{145}{6}z + \frac{57885767}{12} & -\frac{13496555}{8} \\
-5363008 & 3091740 & -\frac{13496555}{8} & \frac{3907}{16}z + \frac{1623646529}{176}
\end{bmatrix}, \allowdisplaybreaks\\
\nn Y^{(4,5)} =& \begin{bmatrix}
\frac{483}{2}z + \frac{68680549}{2} & \frac{78541565}{4} \\
\frac{78541565}{4} & -\frac{255}{8}z + \frac{89966209}{8}
\end{bmatrix} ,
\allowdisplaybreaks\\
\nn Y^{(5,5)} =& \begin{bmatrix}
\frac{2615}{8}z + \frac{310795617}{40} & -\frac{36338387}{8} & -\frac{39201804}{5} \\
-\frac{36338387}{8} & \frac{259}{80}z + \frac{106688481}{40} & \frac{73443691}{16} \\
-\frac{39201804}{5} & \frac{73443691}{16} & \frac{3508}{5}z + \frac{348797983}{44}
\end{bmatrix}, \allowdisplaybreaks\\
\nn Y^{(5,6)} =& \begin{bmatrix}
-\frac{4543}{20}z + \frac{1125579877}{40}
\end{bmatrix},
\allowdisplaybreaks\\
\nn Y^{(6,6)} =& \begin{bmatrix}
\frac{243}{20}z + \frac{69176449}{20} & -\frac{24100707}{4} \\
-\frac{24100707}{4} & \frac{65419}{160}z + \frac{3717945713}{352}
\end{bmatrix},
\allowdisplaybreaks\\
 Y^{(7,7)} =& \begin{bmatrix}
\dfrac{59879}{160}z + \dfrac{2364069}{352}
\end{bmatrix}\,.
\end{align}
\end{center}
\end{widetext}
\twocolumngrid
Above, $z = \sqrt{3}$ and so the entries are in the quadratic number field \(\mathbb{Q}(\sqrt{3})\),
yielding an exact infeasibility certificate.

From $\beta = 1$ it follows that the SDP in Eq.~\eqref{eq:SDP} is infeasible. This ends the proof.
\end{proof}

\section{Symmetry-reduced Lovász}
\label{app:sym_red_Lovasz_app}

\begin{remark}
For the derivations that follow,
we for simplicity
do not impose certain constraints that
appear in the related work on coding bounds Ref.~\cite{munne2024sdpboundsquantumcodes}.
See Appendix~\ref{app:relation} for details.
 \end{remark}

The following derivations are analogous to those of Appendix~\ref{app:inf_cert}.

Let $G_n$ be the anti-commutativity graph of $\mathcal{E}_n \backslash \one$.
Recall the averaging of Eq.~\eqref{eq:app-gamma-twirled} via permutation operators $\pi \in \Pi$.
Via this averaging, define a
symmetry-relaxed Lovász number,
\begin{align}\label{eq:SDP_Lovasz_sym}
\nn \vartheta_{\text{sym}}(G_n)=\text{max} \quad   & \sum_{a \in G_n} M_{aa} \nn\\
\text{s.t.} \quad
& M_{aa} = x_a \quad \quad \forall a\in G_n\nn\\
& M_{ab} = 0 \quad \quad\,\,\,\, \text{if}
\quad a\sim b \nn\\
&\frac{1}{|\Pi|}
\sum_{\pi \in \Pi} \pi
\begin{pmatrix}
1 & x^T \\
x & M
\end{pmatrix}\pi^{-1} \succeq 0
\end{align}
Note that $\vartheta_{\text{sym}}$
is also well-defined
for the Pauli graph
$G_{n, \delta}$
\footnote{For simplicity we used $G_{4+} = G_{7,4}$ in the main text.},
that has all vertices $a$ with $\wt(E_a) < \delta$ removed.
For $\vartheta_{\text{sym}}(G_{n, \delta})$
one thus adds the additional constraint
\begin{equation}
    x_{a}=0 \quad \text{if} \quad 0<\wt(E_a)<\delta  \,.
\end{equation}
in Eq.~\eqref{eq:SDP_Lovasz_sym}.
This sets all rows and columns of $M$
with $0<\wt(E_a)<\delta$ to $0$.

For any $G_{n}$, Eq.~\eqref{eq:SDP_Lovasz_sym} can be block-diagonalized
using the Terwilliger algebra from Appendix~\ref{app:inf_cert}, yielding

\begin{align}\label{eq:sdp_lovasz_relax}
&\mkern-56mu
\vartheta_{\text{sym}}(G_{n}) =  \nn\\
	\textnormal{max} \quad & \sum^n_{i=1} \gamma^{i,i}_{i,i}x^{i,i}_{i,i}
    \nn \\
	\textnormal{s.t.} 	\quad
	&  x^{0,0}_{0,0}=1\,, \nn\\
    & x^{0,0}_{i,0}=x^{i,i}_{i,i}, \quad \text{for} \quad i=1,\dots, n\,, \nn \\
    &x^{t,p}_{i,j} = 0\quad\quad\quad \text{if}\quad  t-p \quad  \text{is odd}\,,  \nn \\
    &\!\!\!\!\!\!\!\!\!\!\!\!\!
    \bigoplus_{\substack{a,k \in \N_0\\ 0\leq a\leq k\leq n+a-k}} \left(\sum\limits_{\substack{t,p\in \N_0 \\ 0 \leq p \leq t \leq i,j \\ i+j\leq t+n}}\alpha(i,j,t,p,a,k)x_{i,j}^{t,p}\right)_{i,j=k}^{n+a-k}   \!\!\!\!\!\!\!\!\succeq 0\,.
\end{align}

For $\vartheta_{\text{sym}}(G_{n, \delta})$ one adds in Eq.~\eqref{eq:sdp_lovasz_relax}
the constraint,
\begin{equation}
x^{t,p}_{i,j} = 0\quad \text{if}\quad  0<i < \delta \quad \text{or} \quad 0<j <\delta\,.
\end{equation}

Numerically we obtain $\vartheta_{\text{sym}}(G_{7,4}) \approx 126.8876$.

\subsection{Infeasibility certificate}
\label{app:inf_cert_Lovasz}
We also provide an infeasibility certificate for
Observation~\ref{obs:lovasz_infeasible}.
The feasibility problem reads as follows:

\bigskip
\noindent\underline{primal:}
\begin{align}\label{eq:SDP_Lovasz_sym_feas}
\text{find} \quad   & M_{aa}\nn\\
\text{s.t.} \quad
& M_{aa} = x_a \quad \quad \forall a\in G_{7}\,,\nn\\
& M_{ab} = 0 \quad \quad\,\,\,\, \text{if}
\quad a\sim b\,, \nn\\
&
\frac{1}{|\Pi|}
\sum_{\pi \in \Pi} \pi
\begin{pmatrix}
1 & x^T \\
x & M
\end{pmatrix}
\pi^{-1}
\succeq 0\,,\nn\\
&1 + \sum_{a\in G_{7}} M_{aa} = 2^7 \,, \nn\\
& M_{aa} = 0 \quad \quad\,\,\,\, \text{if}\quad
0<\wt(E_a) < 4
\,.
\end{align}
Here, the last constraint effectively
restricts to the truncated
Pauli anti-commutativity graph $G_{7,4}$.

The dual of Eq.~\eqref{eq:SDP_Lovasz_sym_feas} reads:

\noindent\underline{dual:}
\begin{align}\label{eq:SDP_Lovasz_sym_feas_dual}
    &\mkern-68mu \beta^* \,  =    \nn\\
    \max_{Y^{(a,k)}\,,\, {w}} \quad
    &   (2^n-1)w - y_{0,0}^{0,0} \nn \\
    \text{s.t.}  \quad & Y^{(a,k)} \succeq 0 \,,\nn \\
    & y_{i,i}^{i,i} + \gamma_{i,0}^{0,0} w + 2 y_{i,0}^{0,0} =0
    \quad\quad\text{if} \quad \delta \leq i \leq n \,, \nn \\
    &  y^{t,p}_{i,j}=0
    \hspace{8.8em}  \text{if} \quad i,j\neq 0\,\, \text{and} \nn\\ \quad &\hspace{4.7cm} t-p \quad  \text{is even,}
   \end{align}
where
\begin{align}
    y^{t,p}_{i,j} = \frac{1}{\gamma^{t,p}_{i,j}}\sum^{\min(i,j)}_{k=0} \sum^k_{\substack{a=\max(i,j)\\\,+k-n}}  \alpha(i,j,t,p,a,k) Y^{(a,k)}_{i-k,j-k}\,.
\end{align}
This can also be seen by modifying the SDP from~\cite[Proposition 26]{munne2024sdpboundsquantumcodes}.
The dual certificate proves that the primal SDP~\eqref{eq:SDP_Lovasz_sym_feas} is infeasible.
It reads:
\begin{widetext}
\begin{center}
\scriptsize
\renewcommand{\arraystretch}{1.25}      
\setlength{\arraycolsep}{2pt}
\begin{align}
\nn Y^{(0,0)} = &\begin{bmatrix}
48397z + 2191397 & 0 & 0 & 0 & -1710z - 158178 & -158290z - 5143 & -5130z - 474153 & -475506z - 15752 \\
0 & -15z + 768721 & 0 & 0 & 0 & 0 & 0 & 0 \\
0 & 0 & -2z + 54796 & 0 & 0 & 0 & 0 & 0 \\
0 & 0 & 0 & -2z + 18090 & 0 & 0 & 0 & 0 \\
-1710z - 158178 & 0 & 0 & 0 & 2z + 11907 & 19359 & 33689 & \frac{412150}{7} \\
-158290z - 5143 & 0 & 0 & 0 & 19359 & z + 36326 & 58411 & \frac{3391382}{35} \\
-5130z - 474153 & 0 & 0 & 0 & 33689 & 58411 & -z + 107704 & 176386 \\
-475506z - 15752 & 0 & 0 & 0 & \frac{412150}{7} & \frac{3391382}{35} & 176386 & -129z + 353284
\end{bmatrix}, \allowdisplaybreaks\\
\nn Y^{(0,1)} =&\begin{bmatrix}
-630z + 32286282 & 0 & 0 & 0 & 0 & 0 \\
0 & -\frac{191}{15}z + \frac{165037913}{150} & 0 & 0 & 0 & 0 \\
0 & 0 & 7z + 131987 & 0 & 0 & 0 \\
0 & 0 & 0 & 15z + 10634 & 13396 & -\frac{433017}{16} \\
0 & 0 & 0 & 13396 & 13z + 18704 & -22040 \\
0 & 0 & 0 & -\frac{433017}{16} & -22040 & -z + 168104
\end{bmatrix} , \allowdisplaybreaks\\
 \nn Y^{(0,2)} =& \begin{bmatrix}
\frac{496}{3}z + \frac{157468426}{15} & 0 & 0 & 0 \\
0 & \frac{629}{8}z + \frac{11474149}{12} & 0 & 0 \\
0 & 0 & \frac{23189}{9144}z + \frac{1399285399}{9144} & -\frac{4543859}{72} \\
0 & 0 & -\frac{4543859}{72} & -34z + 349533
\end{bmatrix}, \allowdisplaybreaks\\
\nn Y^{(0,3)} =& \begin{bmatrix}
-\frac{73}{4}z + \frac{8620129}{2} & 0 \\
0 & -\frac{2266759}{6096}z + \frac{795340555}{6096}
\end{bmatrix}, \allowdisplaybreaks\\
\nn Y^{(1,1)} =& \begin{bmatrix}
21z + 24281583 & 0 & 0 & 0 & 0 & 0 & 0 \\
0 & 26z + 968044 & 0 & 0 & 0 & 0 & 0 \\
0 & 0 & 3z + 211046 & 0 & 0 & 0 & 0 \\
0 & 0 & 0 & 11z + 16138 & -2996 & -\frac{3226851}{160} & \frac{133145}{6} \\
0 & 0 & 0 & -2996 & 27z + 4497 & 11978 & 426 \\
0 & 0 & 0 & -\frac{3226851}{160} & 11978 & 47z + 56335 & 25927 \\
0 & 0 & 0 & \frac{133145}{6} & 426 & 25927 & -160z + 180769
\end{bmatrix}, \allowdisplaybreaks\\
\nn Y^{(1,2)} =&\begin{bmatrix}
-\frac{485}{3}z + \frac{216338887}{6} & 0 & 0 & 0 & 0 \\
0 & -\frac{6259}{32}z + \frac{85570187}{16} & 0 & 0 & 0 \\
0 & 0 & 83z + 182363 & -\frac{18256657}{144} & 212274 \\
0 & 0 & -\frac{18256657}{144} & 15z + 548635 & -\frac{297369}{8} \\
0 & 0 & 212274 & -\frac{297369}{8} & -391z + 833354
\end{bmatrix} , \allowdisplaybreaks\\
\nn Y^{(1,3)} =&\begin{bmatrix}
-\frac{1539}{4}z + \frac{66843657}{2} & 0 & 0 \\
0 & -\frac{320819}{2032}z + \frac{1260614843}{2032} & \frac{51181307}{48} \\
0 & \frac{51181307}{48} & \frac{84357}{1016}z + \frac{1876929723}{1016}
\end{bmatrix} ,\allowdisplaybreaks\\
\nn Y^{(1,4)} =& \begin{bmatrix}
-\frac{765469}{2032}z + \frac{5293363377}{2032}
\end{bmatrix}, \allowdisplaybreaks\\
\nn Y^{(2,2)} =& \begin{bmatrix}
-250z + 15129991 & 0 & 0 & 0 & 0 & 0 \\
0 & 20z + 1674750 & 0 & 0 & 0 & 0 \\
0 & 0 & 20z + 114337 & 58991 & \frac{15671193}{200} & -242540 \\
0 & 0 & 58991 & 48z + 193392 & -175032 & -14153 \\
0 & 0 & \frac{15671193}{200} & -175032 & 169z + 460597 & -16259 \\
0 & 0 & -242540 & -14153 & -16259 & 336z + 1342446
\end{bmatrix},
\allowdisplaybreaks\\
\nn Y^{(2,3)} =& \begin{bmatrix}
-\frac{941}{2}z + 36028913 & 0 & 0 & 0 \\
0 & -\frac{21665}{381}z + \frac{1242931567}{381} & \frac{9746170}{27} & -\frac{23851269}{10} \\
0 & \frac{9746170}{27} & \frac{180493}{381}z + \frac{425682757}{762} & \frac{3750701}{4} \\
0 & -\frac{23851269}{10} & \frac{3750701}{4} & -\frac{103225}{254}z + \frac{579195239}{127}
\end{bmatrix} ,
\allowdisplaybreaks\\
\nn Y^{(2,4)} =& \begin{bmatrix}
-\frac{289023}{1016}z + \frac{3810527075}{1016} & -\frac{232393223}{36} \\
-\frac{232393223}{36} & \frac{223307}{2032}z + \frac{22599297133}{2032}
\end{bmatrix} , \allowdisplaybreaks\\
\nn Y^{(3,3)} =& \begin{bmatrix}
-1241z + 7876176 & 0 & 0 & 0 & 0 \\
0 & -706z + 621766 & -519329 & \frac{56119}{8} & 1636725 \\
0 & -519329 & -99z + 974761 & 625376 & -846248 \\
0 & \frac{56119}{8} & 625376 & 214z + 746408 & 633815 \\
0 & 1636725 & -846248 & 633815 & 3070z + 4832675
\end{bmatrix},
\allowdisplaybreaks\\
\nn Y^{(3,4)} =& \begin{bmatrix}
\frac{95475}{127}z + \frac{2397769627}{127} & \frac{77846069}{12} & 11253942 \\
\frac{77846069}{12} & \frac{754097}{2032}z + \frac{23060948773}{2032} & -\frac{52184735}{8} \\
11253942 & -\frac{52184735}{8} & -236z + 18560512
\end{bmatrix},\allowdisplaybreaks\\
\nn Y^{(3,5)} =&\begin{bmatrix}
\frac{771565}{2032}z + \frac{10555911}{2032}
\end{bmatrix} ,
\allowdisplaybreaks\\
\nn Y^{(4,4)} =& \begin{bmatrix}
-\frac{77753}{127}z + \frac{1055458113}{127} & -\frac{6765269}{3} & -3287901 & -\frac{19878200}{3} \\
-\frac{6765269}{3} & -\frac{775491}{1016}z + \frac{1317891421}{1016} & -\frac{16511057}{24} & \frac{7695517}{2} \\
-3287901 & -\frac{16511057}{24} & -\frac{684341}{762}z + \frac{7589200553}{1524} & -\frac{8530795}{4} \\
-\frac{19878200}{3} & \frac{7695517}{2} & -\frac{8530795}{4} & \frac{1187665}{254}z + \frac{11625442815}{1016}
\end{bmatrix}, \allowdisplaybreaks\\
\nn Y^{(4,5)} =& \begin{bmatrix}
\frac{8451}{254}z + \frac{3892126049}{127} & \frac{140268379}{8} \\
\frac{140268379}{8} & \frac{160819}{1016}z + \frac{5103228305}{508}
\end{bmatrix},
\allowdisplaybreaks\\
\nn Y^{(5,5)} =& \begin{bmatrix}
-\frac{1496877}{1016}z + \frac{8444637027}{1016} & -\frac{77841717}{16} & -\frac{83654767}{10} \\
-\frac{77841717}{16} & -\frac{2951731}{2032}z + \frac{5812362941}{2032} & \frac{39228835}{8} \\
-\frac{83654767}{10} & \frac{39228835}{8} & \frac{1252863}{254}z + \frac{4280039289}{508}
\end{bmatrix}, \allowdisplaybreaks\\
\nn Y^{(5,6)} =& \begin{bmatrix}
-\frac{51637}{508}z + \frac{26226758375}{1016}
\end{bmatrix},
\allowdisplaybreaks\\
\nn Y^{(6,6)} =&\begin{bmatrix}
-\frac{459003}{508}z + \frac{527591631}{127} & -\frac{14506683}{2} \\
-\frac{14506683}{2} & \frac{238700}{127}z + \frac{25842476277}{2032}
\end{bmatrix} ,
\allowdisplaybreaks\\
 Y^{(7,7)} =& \begin{bmatrix}
\frac{328154}{127}z + \frac{1007657}{2032}
\end{bmatrix}\,.
\end{align}
\end{center}
\end{widetext}
\twocolumngrid

\subsection{$\vartheta_{\text{sym}} = \vartheta$
for the Pauli graph}
\label{app:exact}

\begin{proposition}
For $G_{n,\delta}$ the anti-commutativity graph of $n$-qubit Pauli strings
containing only vertices
$a$ with $\wt(E_a) \geq \delta$,
\begin{equation}
 \vartheta_{\text{sym}}\big(G_{n,\delta}\big) = \vartheta\big(G_{n,\delta}\big) \,.
\end{equation}

\end{proposition}
\begin{proof}
    Note that $G_{n,\delta}$
    is invariant under the following operations:
    \begin{enumerate}
    \item The permutation of tensor factors.
    \item
    The permutation of the three non-identity
    Pauli operators
    at each tensor factor independently.
    \end{enumerate}
    One sees that under this operation,
    any feasible point of SDP~\eqref{eq:lovasz} remains feasible, with the objective value $\sum_{a\in G_{n,\delta}} M_{aa}$ for $\vartheta$
    unchanged.
    Vice versa, after symmetrizing any feasible point of the
    SDP~\eqref{eq:SDP_Lovasz_sym},
    \begin{equation}
        \frac{1}{|\Pi|}
\sum_{\pi \in \Pi} \pi
\begin{pmatrix}
1 & x^T \\
x & M
\end{pmatrix}\pi^{-1}\,,
    \end{equation}
    one gets a feasible point for the non-symmetrized SDP~\eqref{eq:lovasz},
    also with its objective value unchanged.
    This ends the proof.
    \end{proof}

    Due to similar reasons, namely the fact that the state $\eEight$ is invariant under permutation of the seven tensor factor pairs $(A_iB_i)$ and under independently permuting the non-identity Pauli matrices of each pair,
    also the relaxation in Eq.~\eqref{eq:sdpx_relax} is exact.

\section{Relation to {\it SDP bounds on quantum codes}}
\label{app:relation}
Here we explain the relation of this work to the
article by the second and last author
{\it SDP bounds on quantum codes}~\cite{
munne2024sdpboundsquantumcodes} and the follow up~\cite{munne2026sdpboundsquantumcodesrational}.
Indeed, all the SDP formulations in this work look deceptively similar to those in Refs.~\cite{munne2024sdpboundsquantumcodes,munne2026sdpboundsquantumcodesrational}.

Let us sketch the relation:
First, on a conceptual level,
\cite{munne2024sdpboundsquantumcodes} is not concerned with entanglement detection, and in fact it is not clear how the machinery would apply to such scenario.
It is exactly the aim of this present paper is to explain how this machinery can be made to apply to entanglement detection,
by establishing a link with the symmetric extension hierarchy and entanglement criteria arising from
the Lovász theta number.

Second, we chose for simplicity in presentation
to omit certain natural constraints that appear in Ref.~\cite{munne2024sdpboundsquantumcodes} and which indeed also would apply to the SDPs here.
In turn, we also {\it added} a constraint which arises from knowing all individual correlations $\av{E_a}\av{E_a}$
of $\eEight$, and not just the value of their sum.
This slightly changes all primal and dual SDPs in this work [Eqs.~\eqref{obs:SDP}, \eqref{eq:SDP_Lovasz_feas}, \eqref{eq:sdpx_relax}, \eqref{eq:dual_sol}, \eqref{eq:SDP_Lovasz_sym_feas}, \eqref{eq:SDP_Lovasz_sym_feas_dual}], compared to~\cite{munne2024sdpboundsquantumcodes}.
Somewhat surprisingly, infeasibility can still be established with these variants.

\bigskip
Let us explain this in more detail. From Eq.~\eqref{eq:EntriesVector} it follows that
 \begin{equation}
     \tr\big((E_{a}\ot E_{a}^{\dag})\Phi_{E8}\big) = 0 \quad\quad \text{if} \quad 0<\wt(E_{a})<\delta\,.
 \end{equation}
 If the state $\Phi_{E8}$ is separable, then by Eq.~\eqref{eq:Gamma_correl},
 \begin{equation} \sum_{i}p_{i}|\tr(E_{i}\mu_{i})|^{2} = 0 \quad   \text{if}  \quad 0<\wt(E_{a})<\delta\,.
 \end{equation}
 It follows that for every $\mu_i$,
 \begin{equation}
 \label{eq:muExpectations}
     \tr(E_{a}\mu_{i}) = 0\quad\quad \text{if} \quad 0<\wt(E_{a})<\delta\,.
 \end{equation}
As a consequence,
\begin{equation}
 \Gamma_{ab} = 0 \quad\quad \text{if}\quad    0<\wt(E_{a}^{\dag}E_{b})<\delta\,.
\end{equation}

We obtain the following additional constraint on the entries of the real part of the moment matrix

\begin{align}
    \label{eq:SDPconstraint}
   \nn \bGamma_{ab} = 0 \quad \text{if}& \quad E_{a}^{\dag}E_{b} = - E_{b}E_{a}^{\dag} \quad \text{or}\, \\
  \nn  &  \quad 0<\wt(E_{a})<\delta \quad \text{or}\quad 0< \wt(E_{b})<\delta  \\
    & \quad \text{or} \quad 0<\wt(E_{a}^{\dag}E_{b})< \delta \,,
\end{align}
where in the present formulation $\delta = 4$. Notice that the constraint from Eq.~\eqref{eq:SDPconstraint} directly corresponds to the fourth condition from~\cite[Eq.~(145)]{munne2024sdpboundsquantumcodes}, where the relationship was encoded in the edges of the confusability graph.

Considering the additional constraint, one needs to modifies also the symmetry reduced  SDP from Eq.~\eqref{eq:sdpx_relax}. Instead of the condition $x_{i,j}^{t,p}=0$ of $t-p$ is odd, one gets $x_{i,j}^{t,p}=0$ of $t-p$ is odd, or $0 < i <\delta $, or $0<j <\delta$, or $0<i+j-t-p<\delta$. This restriction allows one to connect the present formulation with the one from~\cite[Eq.(146)]{munne2024sdpboundsquantumcodes}. Analogously, also the constraint in the dual program in Eq.~\eqref{eq:DualSDPdef} changes from $y_{i,j}^{t,p}=0$ if $i,j\neq 0$ and $t-p$ is even, to $y_{i,j}^{t,p}=0$ if $i,j\neq 0$ and $t-p$ is even and $i+j-t-p \geq \delta$.

We also note that in Ref.~\cite{munne2024sdpboundsquantumcodes},
the notation $i\sim j$ refers to an augmented graph which includes distance relations, whereas in the present work we reserve this notation exclusively to denote anti-commutativity of the corresponding vertex operators.

\bibliography{current_bib}
\end{document}